\documentclass[runningheads]{llncs}

\usepackage[utf8]{inputenc}
\usepackage{verbatim}
\usepackage[ruled,vlined,linesnumbered]{algorithm2e}
\usepackage{amsthm}
\usepackage{graphicx}
\usepackage{amsmath}
\usepackage{xcolor} 
\usepackage{soul} 

\usepackage[hidelinks]{hyperref}

\usepackage{enumitem}
\setlist[description]{topsep=1pt, parsep=1pt, itemsep=1pt, labelindent=.5cm,leftmargin=*}

\title{Consensus Through Knot Discovery in Asynchronous Dynamic Networks} 

\author{Rachel Bricker \and Mikhail Nesterenko \and Gokarna Sharma}
\institute{Kent State University, Kent, OH, 44242, USA\\
\email{rbricke2@kent.edu}, \email{mikhail@cs.kent.edu}, and \email{gsharma2@kent.edu}}

\begin{document}

\maketitle
\begin{abstract}
We state the Problem of Knot Identification as a way to achieve consensus in dynamic networks. The network adversary is asynchronous and not oblivious. 
The network may be disconnected throughout the computation. 
We determine the necessary and sufficient conditions for the existence of a solution to the Knot Identification Problem: the knots must be observable by all processes and the first observed knot must be the same for all processes. We present an algorithm \emph{KIA} that solves it. We conduct \emph{KIA} performance evaluation. 
\end{abstract}
\section{Introduction}

In a dynamic network, the topology changes arbitrarily from one state of the computation to the next. Thus, it is one of the most general models for mobile networks. Moreover, these intermittent changes in topology may represent message losses. Hence, dynamic networks are a good model for an environment with low connectivity or high fault rates.

In such a hostile setting, the fundamental question of consensus among network processes is of interest. 
One approach to consensus is to require that processes in the network remain connected and mutually reachable long enough for them to exchange information and come to an agreement.
However, this may be too restrictive. This is especially problematic if the network is asynchronous and there is no bound on the communication delay between processes. 

The question arises whether it is possible to achieve consensus under less stringent connectivity requirements. In the extreme case, the network is never connected at all. Then, the processes may not rely on mutual communication for agreement.

An interesting approach to consensus is for the processes to use the topological features of the dynamic network itself as a basis for agreement. For example, the process with the smallest identifier or the oldest edge. 
However, as processes collect information about the network topology, due to network delays, such features may not be stable. Indeed, some process may discover another process with a smaller identifier. Therefore, basing consensus decisions on such unstable information may not be possible.

A knot is a strongly connected component with no incoming edges. 
We consider a knot to include edges across some time interval. That is, knot processes may never be connected in a single state.
In general, 
the presence of a knot in a dynamic network is not invariant. 
For example, a dynamic network may have more than one knot, or a knot may grow throughout a network computation. In this paper, we study the conditions under which knots may be used for consensus in asynchronous dynamic networks.

\ \\
\textbf{Related work.}
The impossibility of consensus in asynchronous systems~\cite{flp} in case of a single faulty process has precipitated extensive research on the subject of consensus. Santoro and Widmayer~\cite{santoro1989time} show that consensus is impossible even in a synchronous system subject to link failures. The original paper~\cite{flp} uses knot determination as a topological feature of a computation to achieve consensus. 

There are a number of models related to dynamic networks where consensus is studied. Charron-Bost and Schiper~\cite{charron2009heard} introduce a heard-of (HO) model and consider consensus solvability there. In this model, the set of ``heard-of'' for each process is analyzed. The consensus is proved to be solvable only if there is an agreement between processes on these heard-of sets. 

There is a related research direction of consensus with unknown participants~\cite{cavin2004consensus}, where participant detectors perform a similar role to links in dynamic networks.  
Rather than place restrictions on the dynamic network to enable a solution, Altisen et al.~\cite{altisen2021implementing} relax the problem and consider an eventually stabilizing version of it. 

Kuhn {\it et al.}~\cite{kuhn2010distributed,kuhn2011coordinated} study consensus under a related model of directed networks. In their case, it is assumed that there exists a spanning subgraph of the network within $T$ rounds. In the case $T=1$, the network is always connected.

Afek and Gafni~\cite{afek2013asynchrony} introduce the concept of an adversary as a collection of allowed network topologies. An oblivious adversary~\cite{coulouma2015characterization} composes an allowed computation by selecting the topology of each state from a fixed set of allowed network topologies. There are a number of papers that study consensus under this adversary~\cite{castaneda2021topological,fevat2011minimal,winkler2023time}.

In the case of a non-oblivious adversary, no restrictions on the potential state topologies are placed. 
In this case, the adversary completely controls the connectivity of the network in any state and the changes in connectivity from state to state. 
In the work known to us, to solve consensus
under such a powerful adversary, extra connectivity assumptions are assumed. Biely {\it et al.}~\cite{biely2012agreement} consider consensus under an eventually stabilizing connected root component. Some papers study the case where the network stays connected long enough to achieve consensus~\cite{biely2018gracefully,schwarz2016fast}.

In the present work, we assume a non-oblivious adversary and consider the case where the system may remain disconnected in any state of the computation.

\ \\
\noindent\textbf{Our contribution.} 
We use the non-oblivious adversary defined and studied previously~\cite{biely2012agreement,biely2018gracefully,schwarz2016fast}. We focus on knot identification under such adversary. 
We define the Knot Identification Problem and study it for directed dynamic networks. 
This problem requires the network processes to agree on a single knot.
The solution to this problem can immediately be used to solve consensus.

We assume an asynchronous adversary. The asynchrony of the adversary allows it to delay the communication between any pair of processes for arbitrarily long. If the adversary is asynchronous, each process may not hope to gain additional information by waiting and must make the output decision on the basis of what it has observed so far. 

We consider a \emph{knot observation final} adversary. In such adversary, it is possible that the collection of knots observed by some process may not increase throughout the rest of the computation. Thus, the processes may not ignore any of the observed knots hoping to get others later. Instead, each process must make the decision on the basis of the knots seen so far.

Since the processes must agree on a knot, a process may output the knot only if it is observed by other processes. Hence, the same knot needs to be observed by all processes. Once a process observes a knot, it must determine if this knot is observed by everyone else. We call an adversary \emph{knot opaque} if it does not allow a process to determine whether the knot is observed by other processes or not. We prove that there is no solution to the Knot Identification Problem for such a knot opaque adversary. We then consider adversaries that are \emph{knot transparent} rather than knot opaque. 

For an adversary that is asynchronous, knot transparent, and knot observation final, we prove that it is necessary and sufficient for all processes to observe the same first knot. 

For sufficiency, we present a simple knot identification algorithm \emph{KIA} that solves the Knot Identification Problem. 
We conduct performance evaluation to study \emph{KIA} behavior. This evaluation studies the dynamics of knot detection under various parameters. It demonstrates the practicality of \emph{KIA} and our
approach to consensus for dynamic networks with little connectivity.

\section{Notation and Problem Definitions}

We state the notation to be used throughout the paper in this section. To simplify the exposition, we add further definitions in later sections, closer to place of their usage. 
 
\ \\
\textbf{Links, states, computations.} The network consists of $N$ \emph{processes}. The processes have unique identifiers. No process a priori knows $N$ or the identifiers of the other processes. 

A pair of processes may be connected by a unidirectional \emph{link}. The network \emph{state} $s$ is a collection of such links that, together with the processes, form a \emph{state communication graph} or just \emph{state graph}. 
Thus, processes are nodes in this graph. 
One specific state is a \emph{non-communicating state} whose state graph has no links. That is, all processes are disconnected in this state.

A \emph{computation} $\sigma$ is an infinite sequence of network states. An \emph{adversary} is a set of allowed computations. 
Given an adversary, an algorithm attempts to solve a particular problem. We use the term computation for both the states allowed by the adversary and the operation of the algorithm in these states.

To aid in the solution, processes exchange information across existing links. If two processes are connected by a link in  a particular state, the sender may transmit an unlimited amount of information to the receiver. This communication is reliable.
The sender does not learn the receiver's identifier.

The processes do not fail.  Alternatively, a process failure may be considered as permanent disconnection of the failed process from the rest of the network. 

\ \\
\noindent
\textbf{Causality, asynchrony, computation graphs, and knots.}
A computation \emph{event} is any computation action or topological occurrence that happens in a computation. Examples of computation events are a process carrying out its local calculations or an appearance of a link. 
Given a particular computation, a communication event $e_1$  \emph{causally precedes} another event $e_2$  if (i) both events occur in the same process and $e_1$ occurs before $e_2$; (ii) there is a communication link between processes $p_1$ and $p_2$ and $e_1$ occurs at $p_1$ before the link and $e_2$ occurs at $p_2$ after the link; 
(iii) there is another event $e_3$ such that $e_1$ causally precedes $e_3$ and $e_3$ causally precedes $e_2$. 
We consider the presence of a link in a particular state to be a single event. If the same link is present in the subsequent state, it is considered a separate event. This way, causal precedence is defined for links. 
Note that the insertion of a non-communicating state into a computation preserves all causality relations of the computation. 

Consider a computation $\sigma_1$ allowed by some adversary $\mathcal{A}$. Let $\sigma_2$ be obtained from $\sigma_1$ by inserting a non-communicating state after an arbitrary state of $\sigma_1$. If the adversary $\mathcal{A}$ also allows $\sigma_2$, then $\mathcal{A}$ is  \emph{asynchronous}. Intuitively, an asynchronous adversary may delay process communication for arbitrarily long.

Given a computation $\sigma$, a \emph{computation graph} $G(\sigma,i)$ is the union of all the state graphs up to and including state $s_i$. To put another way, the computation graph is formed by the processes
and the links present in any state $s_j$, for $j \leq i$. 

A \emph{knot} is a strongly connected subgraph with no incoming links. 
A process $p_i$ is in a knot if for every process $p_j$ reachable from $p_i$, $p_i$ is reachable from $p_j$. Given a graph $G$, this definition suggests a simple knot computation algorithm. For each process in $G$, compute a reachability set $S$. For a process $p_i$ with reachability set $S_i$, if there is a process $p_j \in S_i$ such that $p_i \in S_j$ then, $p_i$ and $p_j$ are in the same knot.

When it is clear from the context, we use the term \emph{knot} for both the subgraph and for the set of processes that form this subgraph. Any process that has not communicated yet is trivially a singleton knot. Therefore, we only consider knots of size at least two.  

Computation $\sigma$ contains a knot $K$ if there is a state $s_i$, $i<\infty$, such that $G(\sigma, i)$ contains  $K$.
Note that there is no requirement that the edges of the knot in a computation are causally related, just that the union of all state graphs up to some state $s_i$ contains a knot. 
As the computation progresses, edges are added to the computation graph of this computation. In general, a knot in this graph is not stable. If an incoming edge is added, the knot may disappear. Similarly, added links may expand the knot by joining mutually reachable processes. 

\ \\
\textbf{Observability.}
A \emph{local observation graph} $LG(p, \sigma, i)$ is all the links and adjacent processes that causally precede the events in $p$ in state $s_i$ of computation $\sigma$. A local observation graph $LG(p, \sigma, i)$ is thus a subgraph of the computation graph $G(\sigma, i)$.
In effect, the local observation graph of $p$ is what $p$ sees of the computation so far. 
In the beginning of the computation $LG$ of process $p$ is empty and $LG$ grows as $p$ receives topological information from incoming links.

Two computations $\sigma_1$ and $\sigma_2$ are \emph{observation graph identical} for process $p$ up to state $s_i$ if $LG(p,\sigma_1,i) = LG(p,\sigma_2,i)$.

\begin{figure}[htbp]
   \centering 
   \includegraphics[width=0.5\textwidth]{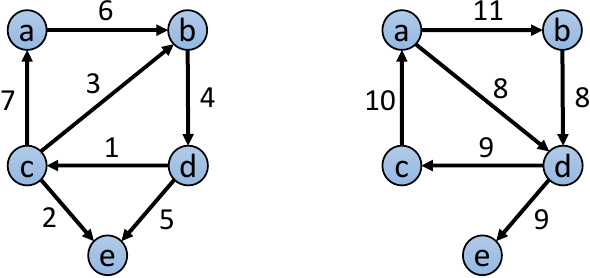}
   \caption{Knot formation example. Edge labels denote states when the edges are present. Process $e$ observes knot $K_1 = \{b,c,d\}$; process $d$ is the first to observe knot $K_2 = \{a,b,c,d\}$. }
   \label{knotPic}
\end{figure}

Let us illustrate these concepts with an example shown in Figure~\ref{knotPic}. In state $4$, a knot $K_1 = \{b,c,d\}$ is formed due to the links $d\rightarrow c$, $c \rightarrow b$ and $b \rightarrow d$. In state $5$, due to the link $d \rightarrow e$, process $e$ observes $K_1$.  
In state $6$, $K_1$ is destroyed because of an incoming link: $a \rightarrow b$. In state $7$, link $c \rightarrow a$ creates knot $K_2 = \{a, b, c, d\}$. In state $8$, process $d$ observes $K_2$. In the remaining states, all processes observe $K_2$.

In general, a knot may exist in the computation graph but may not be visible to any of the processes that belong to this knot or even any of the processes in the network at all. Indeed, 
processes that belong to a knot may not see said knot because it does not belong to their local observation graphs. For example, in Figure~\ref{knotPic}, if no more links appear after state 7
in the computation,
none of the processes in knot $K_2$, or even in the entire network, observe $K_2$, yet it exists in the computation graph.

A knot $K$ is \emph{observable} in  computation $\sigma$ by process $p$ if there is a state $s_i$ such that $K \subset LG(p,\sigma,i)$.  A knot is \emph{globally observable} in a computation if it is observable by every process in the network. 
That is, a knot is globally  observable if every process eventually sees it.

Consider the earliest state in the computation $\sigma$ where $p$ observes knot $K$.
This state contains an incoming link (or links) to $p$ that brings additional topological information to $LG(p, \sigma, i)$ to complete the knot $K$.  This link is the \emph{observation event} at process $p$ for this knot. For example, in Figure~\ref{knotPic}, link $d \rightarrow e$ is the observation event for $K_1$ at process $e$.

\ \\
We consider algorithms that are deterministic in the following way. 
If two computations $\sigma_1$ and $\sigma_2$ are observation graph identical for process $p$ up to state $s_i$, then all the outputs of $p$ up to state $s_i$ for algorithm $\mathcal{S}$ in the two computations are identical.
Put another way, in such an algorithm, each process makes its decisions only on the basis of its local observation graph. 

\ \\
\textbf{The Knot Identification and Consensus Problems.} 

\begin{definition}[Consensus]
Given that every process is input a binary value $v$, a consensus algorithm requires each process to output a decision value following the three properties.

\begin{description}
\item[]{\em Consensus Validity:} if all processes are input the same value $v$, then output decision is $v$;
\item[]{\em Consensus Agreement:} if one process outputs $v$, then every output decision is $v$; 
\item[]{\em Consensus Termination:} every process decides.
\end{description}
\end{definition}

\begin{definition}[Knot Identification]
A solution to the Knot Identification Problem requires that
given a computation, 
each process outputs the set of processes $K$ that form a knot in this computation. The output is subject to the following properties.  
\begin{description}
\item [] \emph{KI Agreement:} if one process outputs a knot $K$, then every output knot is also $K$;
\item [] \emph{KI Termination:} every process outputs a knot.
\end{description}
\end{definition}

An adversary is \emph{consensual} if there exists an algorithm that solves Consensus on every computation allowed by this adversary. Similarly, a \emph{knot-identification} adversary admits an algorithm that solves this problem on every allowed computation. 

Once the Knot Identification Problem is solved, consensus follows. Indeed, if all processes agree on a knot, they may use it to determine the consensus value to be output. For example, the consensus value may be the input to the knot process with the highest identifier, or the process incident to the oldest link, etc. We state this observation in the below proposition. 

\begin{proposition}
A knot-identification adversary is also a consensual adversary.
\end{proposition}

\noindent In the remainder of the paper, we focus on the Knot Identification Problem. 

\section{Necessary and Sufficient Conditions for\\ Knot Identification}

\textbf{Knot opacity.} The KI Agreement property requires that every process outputs the same knot. A process may output only a knot that it observes. Hence, the following proposition. 

\begin{proposition}
In a solution to the Knot Identification Problem, every process outputs only a globally observable knot. 
\end{proposition}

However, even if an adversary has a globally observable knot in every computation, it does not guarantee that this adversary
admits a solution to the Knot Identification Problem. A process observing a particular knot must also know whether or not this specific knot is globally observable.  Let us discuss this in detail.

An adversary $\mathcal{A}$ is \emph{knot opaque} if there is a process $p$ and a computation $\sigma_1 \in \mathcal{A}$ such that for every state $s_i$ of  $\sigma_1$ and every knot $K$ observed by $p$ in states up to $s_i$, there is another computation $\sigma_2$ that is local observation graph identical
to $\sigma_1$ 
for $p$ up to $s_i$, yet $K$ is not globally observable in $\sigma_2$.
Intuitively, a knot opaque adversary does not allow a process $p$ to distinguish whether or not any knot $K$ that $p$ observes is also observed by all other processes, i.e. this knot is globally observable. An adversary is \emph{knot transparent} if it is not knot opaque.

\begin{lemma}\label{opaque}
There does not exist a solution to the Knot Identification Problem for a knot opaque adversary. 
\end{lemma}
\begin{proof}
Assume the opposite. Suppose there exists a knot opaque adversary $\mathcal{A}$. Also, let $\mathcal{S}$ be the algorithm that solves the Knot Identification Problem in $\mathcal{A}$. Since $\mathcal{A}$ is knot opaque, there exists a computation $\sigma_1$ and a process $p_1$ such that 
for every knot that $p_1$ observes, it is unclear to $p_1$ whether or not this knot is globally observable. 

Algorithm $\mathcal{S}$ is assumed to be a solution to the Knot Identification Problem. 
According to the KI Termination property, $p_1$ in $\sigma_1$ must output one of its observed knots. Let $p_1$ output knot $K$ in some state $s_i$ of $\sigma_1$. Since $\mathcal{A}$ is knot opaque, it contains a computation $\sigma_2$ that is observation graph identical
to $\sigma_1$ 
for $p_1$ up to state $s_i$, yet knot $K$ is not globally observable in $\sigma_2$. 
If $\sigma_2$ is observation graph identical to $\sigma_1$
for $p_1$ up to state $s_i$, then process $p_1$ in algorithm $\mathcal{S}$ outputs $K$ in $\sigma_2$ just like it does in $\sigma_1$. 

If knot $K$ is not globally observable in $\sigma_2$, then there is a process $p_2$ that does not observe $K$
in $\sigma_2$. If so, $p_2$ in $\sigma_2$ either outputs a knot different from $K$ or none at all. In the first case, $\mathcal{S}$ violates the KI Agreement property that requires that every process outputs the same knot. In the second case, if $p_2$ does not output a knot in $\sigma_2$, $\mathcal{S}$ violates KI Termination Property requiring every process to output a knot. 

In either case $\mathcal{S}$ does not comply with the properties of the Knot Identification Problem. This means that, contrary to our initial assumption, $\mathcal{S}$ may not be a solution to this problem. Hence the lemma.
\end{proof}

\ \\
\textbf{Knot finality.}
Lemma~\ref{opaque} restricts the adversary from hiding whether a particular knot a process observes is globally observable or not.
However, even if each process knows if the knot is globally observable, it may still be insufficient to ensure the existence of a solution.

Consider an arbitrary computation $\sigma_1$ and an arbitrary process $p$ of some adversary $\mathcal{A}$. An adversary $\mathcal{A}$ is  \emph{knot observation final} if it contains a computation $\sigma_1$ where there is a process $p$ such that, for every state $s_i$
of $\sigma_1$,
there is a computation $\sigma_2$ which is observation graph identical to $\sigma_1$
for $p$ up to state $s_i$
such that, after state $s_i$, it does not contain any more knot observations by $p$.  Intuitively, in such an adversary, a process may not gain additional knot information by delaying its decision.

A knot is \emph{primary} for some process $p$ in computation $\sigma$ if it is the first observed knot by $p$ in $\sigma$.

\begin{figure}[htbp]
   \centering 
   \includegraphics[angle=-90,width=0.8\textwidth]{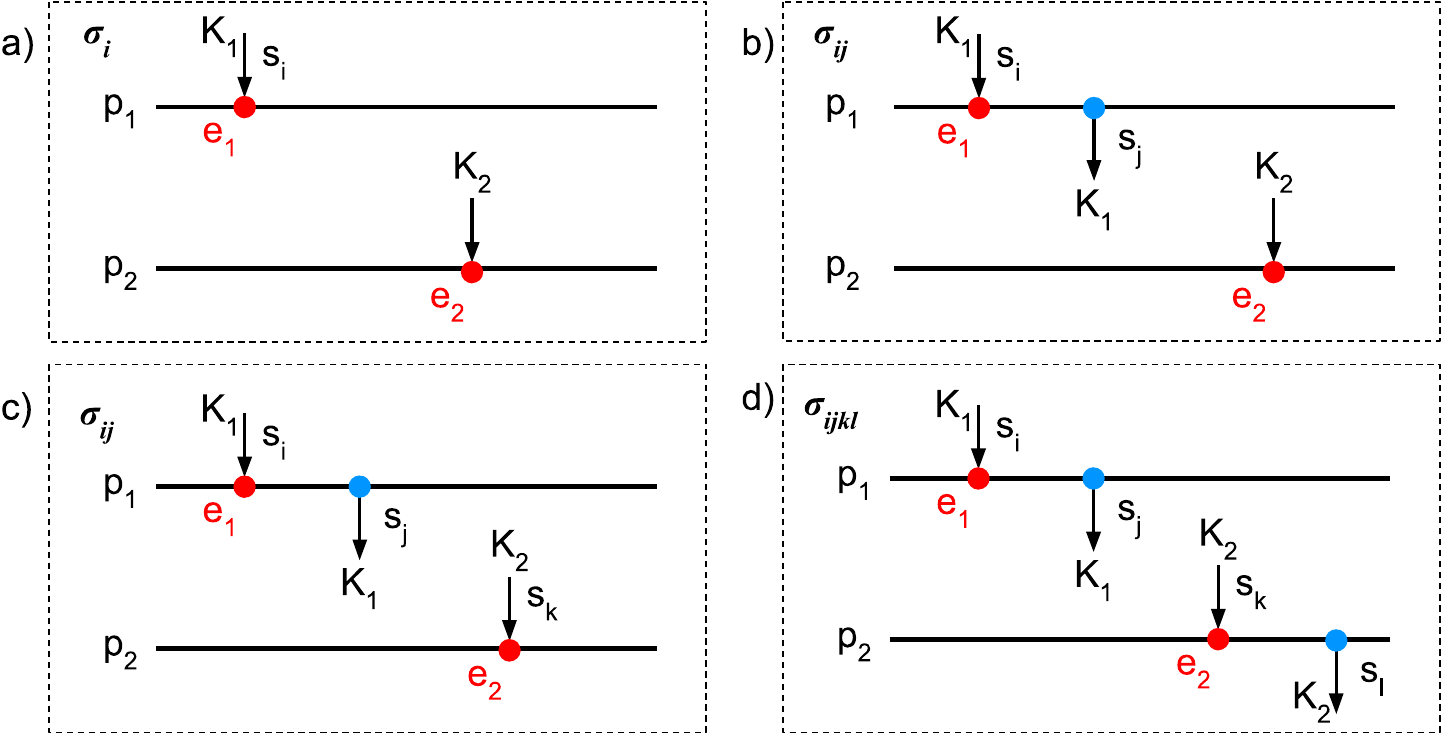}
   \caption{Illustration for the proof of Lemma~\ref{lemNoAsynch}. In figure a), in computation $\sigma_i$, process $p_1$ observes knot $K_1$ with event $e_1$ in state $s_i$. In figure b), in computation $\sigma_{ij}$, process $p_1$ outputs knot $K_1$ in state $s_{ij}$. In figure c), in the same computation $\sigma_{ij}$, process $p_2$ observes knot $K_2$ in state $s_k$ with event $e_2$. In figure d), in computation $\sigma_{ijkl}$, process $p_2$ outputs $K_2$ in state $s_l$.}
   \label{lemma2pic}
\end{figure}

\begin{lemma}\label{lemNoAsynch}
Consider an observation final, asynchronous, knot transparent adversary $\mathcal{A}$. 
If $\mathcal{A}$ contains a computation $\sigma$ such that a pair of processes observe two
different primary knots, then this adversary does not have a knot identification solution even
though this adversary is knot transparent. 
\end{lemma}

\begin{proof}
Consider the adversary $\mathcal{A}$ that conforms to the conditions of the lemma. Yet, there is an algorithm $\mathcal{S}$ that solves the Knot Identification Problem on $\mathcal{A}$. 
According to the lemma conditions, $\mathcal{A}$ allows some computation $\sigma$ with a pair of processes $p_1$ and $p_2$ that observe different primary knots $K_1$ and $K_2$, respectively. Since $\mathcal{A}$ is knot transparent, $K_1$ and $K_2$ may be globally observable. 
Refer to Figure~\ref{lemma2pic} for illustration.

Let $e_1$ and $e_2$ be the corresponding knot observation events in $\sigma$.
The two events may be either 
concurrent or causally dependent. In the latter case, we assume, without loss of generality, that $e_1$ causally precedes $e_2$. Let event $e_1$ occur in state $s_i$. 

Since $\mathcal{A}$ is knot observation final, it allows a computation $\sigma_{i}^{\prime}$ that is observation graph identical for $p_1$ to $\sigma$ up to state $s_i$, yet $p_1$ does not observe any knots after state $s_i$ in $\sigma_{i}^{\prime}$. That is, the only knot $p_1$ observes is $K_1$. (We denote computations where a process does not observe any more knots with the prime symbol.)
%

Since $\mathcal{S}$ is assumed to be a solution to the Knot Identification Problem, each process, including $p_1$, must output a knot in $\sigma_{i}^{\prime}$. The only knot that $p_1$ observes in $\sigma_{i}^{\prime}$ is $K_1$. Hence, $p_1$ outputs $K_1$. It may output it in state $s_i$, or in some later state. We consider the case where $p_1$ outputs $K_1$ later. 

Since $K_1$ is primary for $p_1$, the observation event 
$e_1$ for $K_1$ at $p_1$ in $\sigma$ causally precedes observation events of other knots at $p_1$ if such observations ever happen. We construct a computation $\sigma_{i1}$
from $\sigma$
by adding a non-communication state after state $s_i$. Since $\mathcal{A}$ is an asynchronous adversary, $\mathcal{A}$ allows $\sigma_{i1}$. Note that $\mathcal{A}$ also allows a computation $\sigma_{i1}^{\prime}$  which is observation identical to $\sigma_{i1}$
for $p_1$ for states up  to $s_{i+1}$ but where
$p_1$ observes no other knots besides $K_1$. Similarly, $\mathcal{S}$ must have $p_1$ output $K_1$ in $\sigma_{i1}^{\prime}$. This output occurs in state $s_{i+1}$ or later. 

Note that the purported solution to the Knot Identification Problem $\mathcal{S}$ has to comply with its Termination Property. This means that each process must eventually output a knot. Therefore, as we continue this process of adding non-communication states past $s_i$, we find computation $\sigma_{ij} \in \mathcal{A}$ where $p_1$ outputs $K_1$ in state $s_j$ following state $s_i$.

Let us examine $\sigma_{ij}$. In this computation, $p_2$ observes its primary knot $K_2$ with observation event $e_2$. By construction, $e_2$ happens in some state $s_k$ following $s_j$. Similar to the above procedure, we continue adding non-communication states past $s_k$ until we obtain computation $\sigma_{ijkl}$ where $p_2$ outputs knot $K_2$ in state $s_l$ following state $s_k$. 

Let us now examine $\sigma_{ijkl}$. In this computation, in algorithm $\mathcal{S}$, process $p_1$ outputs knot $K_1$ while process $p_2$ outputs knot $K_2$. However, these two knots are different. Therefore, $\mathcal{S}$ violates the Agreement Property of the Knot Identification Problem requiring every process to output the same knot. Yet, this means that $\mathcal{S}$ may not be a solution to this problem and our initial assumption is incorrect. This proves the lemma.
\end{proof}

An adversary $\mathcal{A}$ is \emph{primary uniform} if the following conditions hold  for every computation $\sigma \in \mathcal{A}$ : (i) each process observes at least one knot; (ii) if some process $p_1$ observes its primary knot $K_1$  and another process observes its primary knot $K_2$, then $K_1 = K_2$. To put another way, in a single computation of a primary uniform adversary, all processes observe the same primary knot. 

\begin{theorem}\label{thrmIFF}
For a knot observation final asynchronous knot transparent adversary $\mathcal{A}$ to allow a solution to the Knot Identification Problem, it is necessary and sufficient for $\mathcal{A}$ to be primary uniform.  
\end{theorem}

\begin{proof}
The necessity part of the theorem follows from Lemma~\ref{lemNoAsynch}. We prove the sufficiency by presenting the algorithm \emph{KIA} below that solves the Knot Identification Problem under $\mathcal{A}$. 
\end{proof}

\section{Knot Identification Algorithm \emph{KIA}}

\noindent
\textbf{Description.} The knot identification algorithm \emph{KIA} operates as follows. See Algorithm~\ref{algKIA}. Across every available outgoing 
link, 
each process $p$ relays all the connectivity data that it has observed so far. That is, if process $p$ communicates with process $q$ at state $s_{i}$ of computation $\sigma$, then $p$ transmits its entire local observation graph, $LG(p,\sigma,i)$, to $q$.

Once a process $p$ detects a knot in $LG(p, \sigma, i)$, it outputs it.
Since the adversary is primary uniform, each process is guaranteed to eventually observe a primary knot and this knot is the same for every process.  That is, \emph{KIA} solves the Knot Identification Problem.

\begin{algorithm}[htbp]
\SetKwComment{Comment}{$\triangleright$\ }{}
\textbf{Constants:} \\
    $p$ \Comment*[f]{process identifier} \\
\vspace{2mm}
\textbf{Variables:} \\
 $LG(p,\sigma,i)$ \Comment*[f]{local observation graph of process $p$} \\
 
\vspace{2mm}
\textbf{Actions:} \\

\If{exist outgoing links}{
   \textbf{send} $LG(p,\sigma,i)$ \emph{to every outgoing link}
}

\If{\textbf{\em receive} $LG(q,\sigma,i)$ from process $q$}{
    $LG(p,\sigma,i) = LG(p,\sigma,i)\cup LG(q,\sigma,i)$ \Comment*[f]{merge graphs}
}

\If{$\exists$ knot\ $K$:\ $K \subset LG(p,\sigma,i)$ }{
    \textbf{output} $K$ \Comment*[f]{report knot}
}

\caption{Knot Identification Algorithm \emph{KIA}}
\label{algKIA}
\end{algorithm}

\ \\
\textbf{Complexity estimation.} Let us estimate the number of states it takes for each process of \emph{KIA} to output its decision. This estimate is tricky since the algorithm may not do anything productive if no edges appear.  Hence, we only count states where information spreads. To put another way, we compute the worst case number of causally related links before every process outputs a knot. 

Let $n$ be the number of processes in the network. The algorithm operation can be divided into two parts: (i) knot formation and (ii) knot data propagation. In the worst case, these two parts run consecutively. Suppose the last process, $p$, that participates in the knot is the first to observe it. Then, $p$ is the only process that informs the other processes of the knot. To put another way, the knot observations at the other processes are causally preceded by the knot observation at $p$.

The knot with the longest causally related links is a cycle of $n$ edges. 
The knot data propagation part requires $n-1$ edges if all processes are informed sequentially. Hence, the worst case \emph{KIA} complexity is $2n-1$, which is in $O(n)$.

\section{\emph{KIA} Performance Evaluation}
\begin{figure}[htb]
   \centering 
   \includegraphics[width=0.6\textwidth]{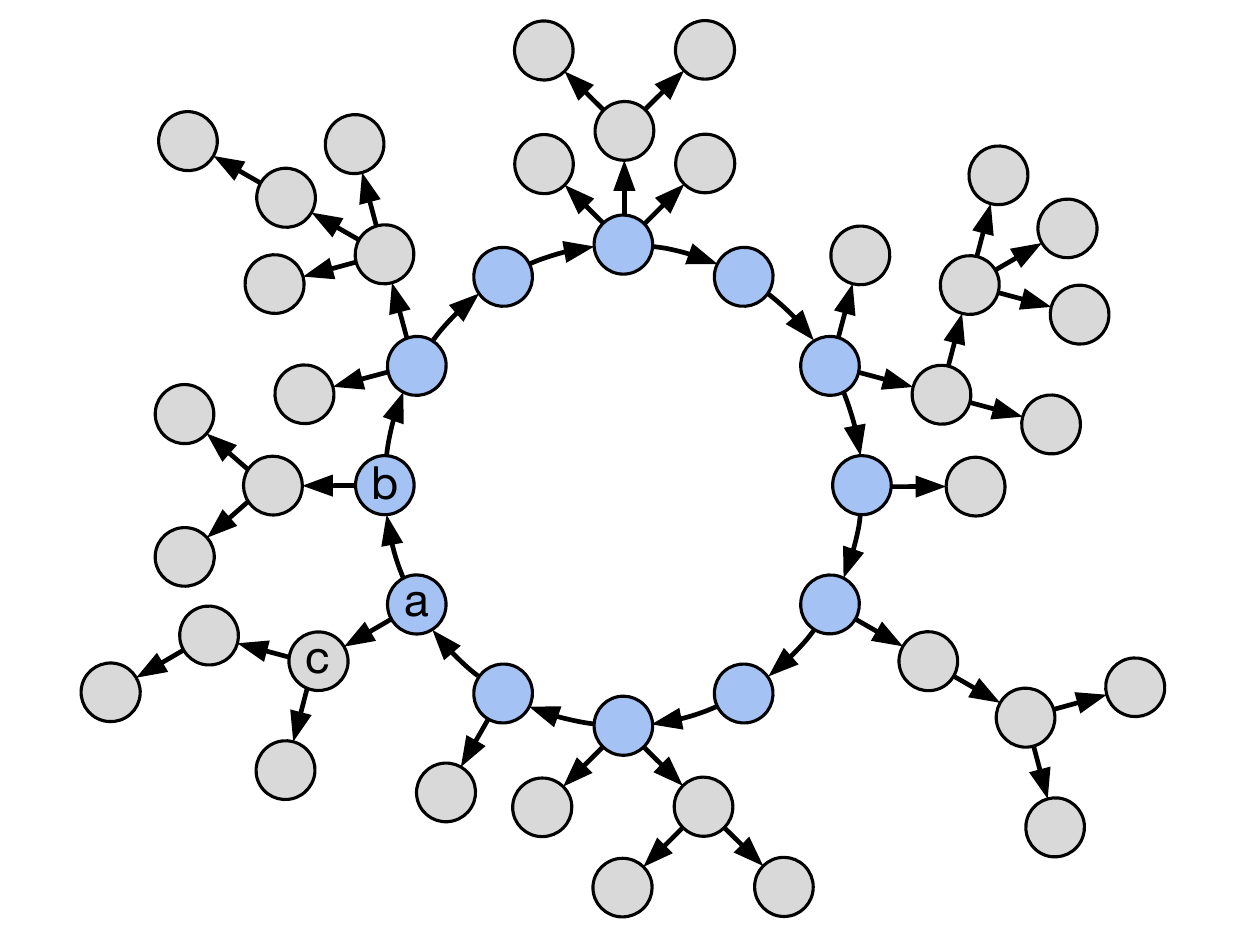}
   \caption{Intermittent Connectivity Topology Example. The underlying topology contains a single knot: the cycle.}
   \label{figCycleOfTrees}
\end{figure}

\begin{figure}[htbp]
\centering
    \includegraphics[width=8cm]{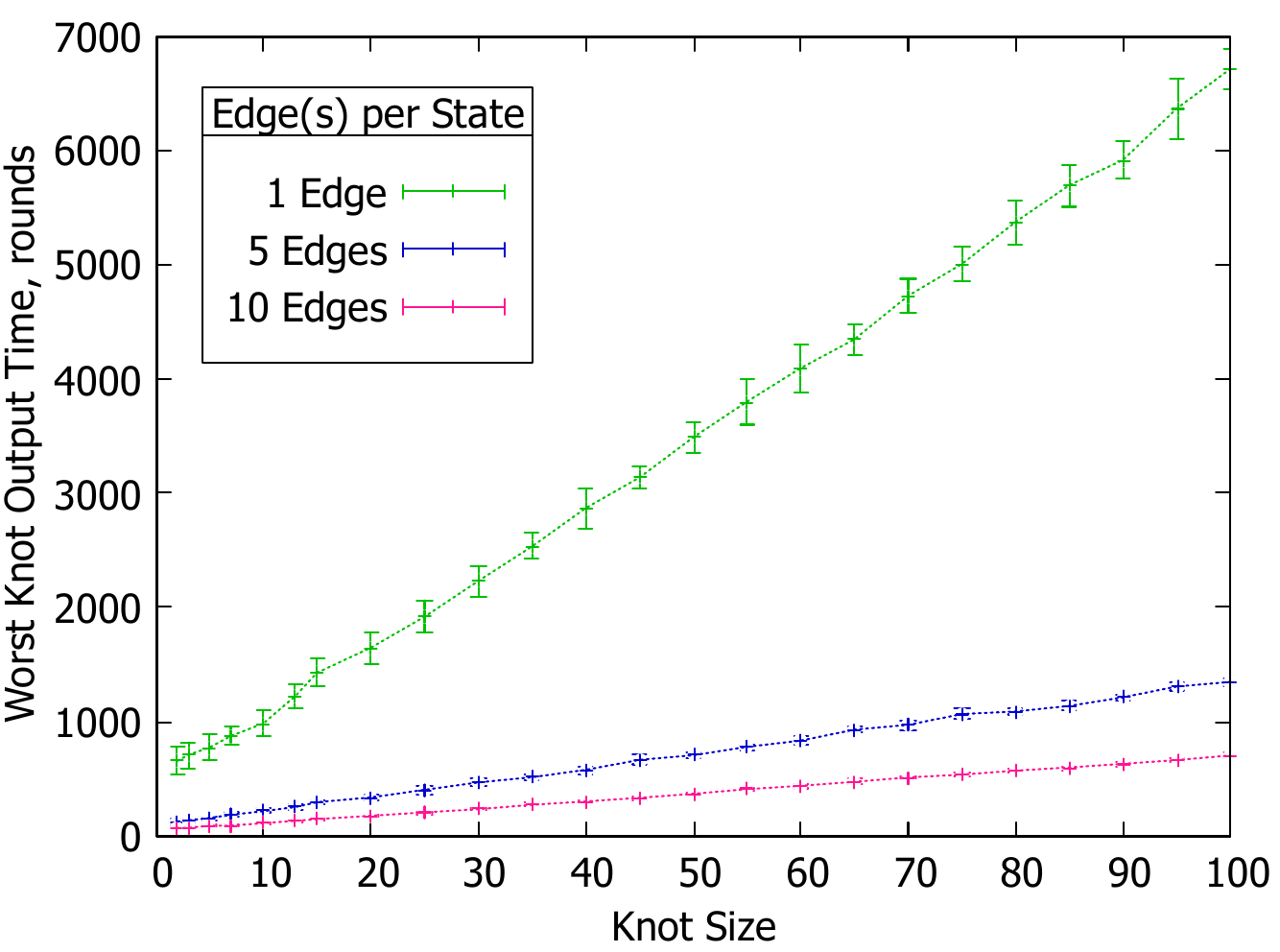}
   \caption{Longest knot output time as a function of the knot size.}
   \label{varyingKnotSizeGraph}
\end{figure}

\begin{figure}[htb]
\centering
\includegraphics[width=8cm]{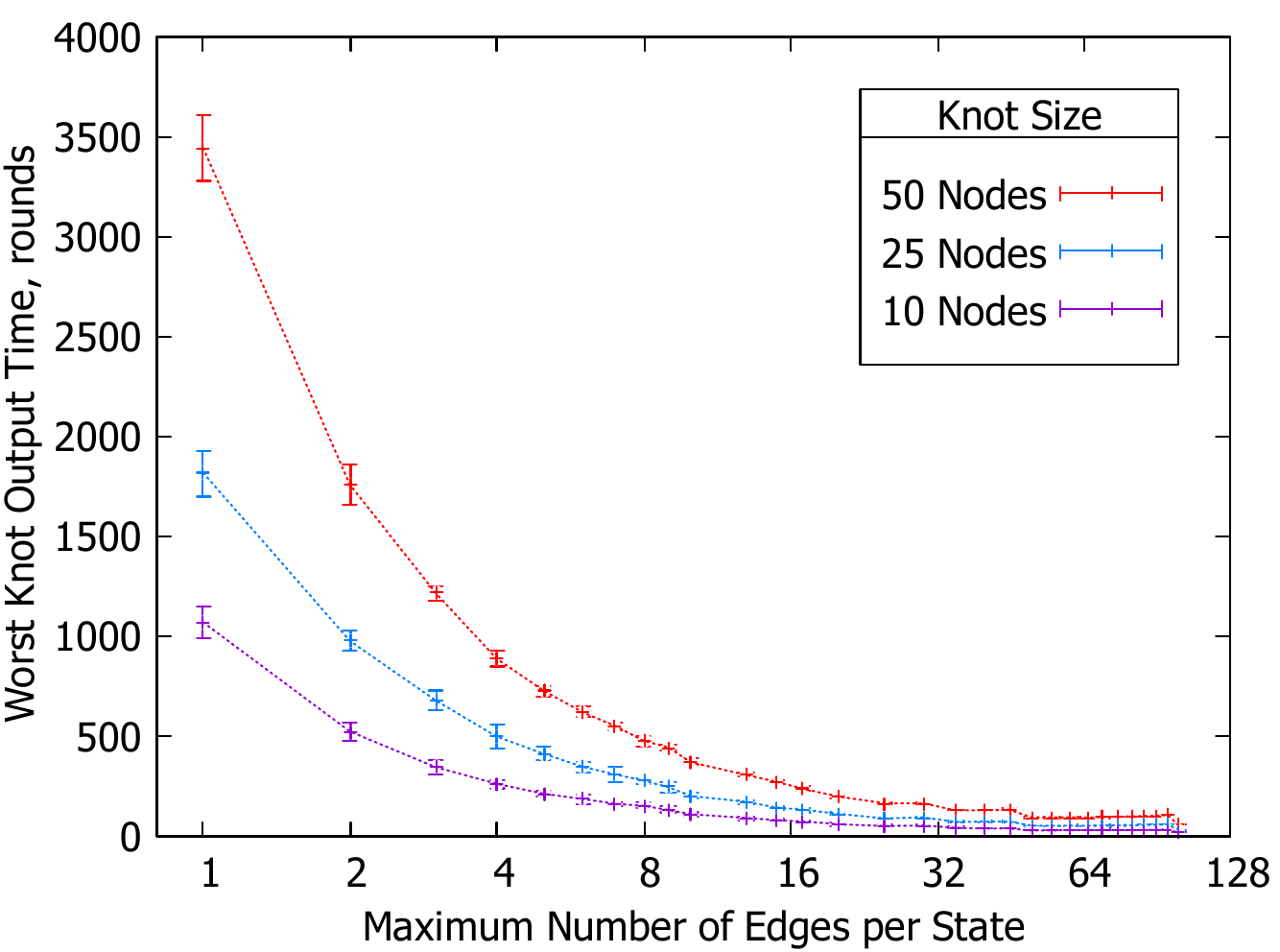}
   \caption{Longest knot output time as a function of the maximum number of edges per state.
   }
   \label{varyingEdgesPerStateGraph}
\end{figure}

We studied the performance of our Knot Identification Algorithm \emph{KIA} using an abstract algorithm simulator QUANTAS~\cite{oglio2022quantitative}. The QUANTAS code for \emph{KIA} as well as our performance evaluation data is available online~\cite{KiaGitHub,KiaData}.
The computations were selected as follows. First, we generated the underlying \emph{backbone} topology. In the backbone, a certain number of nodes are jointed in a cycle. Each remaining node is randomly attached with a single edge to an already selected connected node.
See Figure~\ref{figCycleOfTrees} for an example of such a topology. In each round of a computation, a fixed number of backbone edges appear. The edges to appear are selected uniformly at random. Thus, each computation contains a single knot---the backbone cycle---while the whole network is unlikely to be connected in a single round. Moreover, the information about this cycle is eventually propagated to all nodes in the network. That is, all generated computations contain exactly one globally observable knot. 

We implemented \emph{KIA} and measured its performance. We measured the speed of knot detection expressed as the longest number of rounds it takes for any process in the network to output the knot. 

In the first experiment, we fixed the number of random edges appearing per round and varied
the knot (cycle) size. The results are shown in Figure~\ref{varyingKnotSizeGraph}. We set the network size to $100$ nodes. The cycle size varies from $2$ to $100$. That is, the largest cycle comprises the whole network. The computation length is set at $6,000$ rounds. The knot output time is averaged across $10$ computations. We plot \emph{KIA} performance for the case of $1$, $5$ and $10$ backbone edges appearing per round. The data shows that smaller knots are detected quicker by all the nodes in the network. 

In the second experiment, we fixed the knot, i.e. cycle, size and varied the number of random edges per round.  The results are shown in Figure~\ref{varyingEdgesPerStateGraph}. Intuitively, it shows that a greater number of edges appearing in one round, even if the network remains disconnected, provides greater overall connectivity and accelerates knot detection. 

\ \\
Our experiments demonstrate the practicality of our knot identification approach to agreement in dynamic networks.

\section{Extensions of Knot-Based Consensus}

\textbf{Distinguished knots.} In this paper, we treated the problem of knot-based consensus as generally as possible. However, it may be adapted to particular systems: certain topologies may be significant to the system and the processes could be programmed to distinguish such knots. For example, the processes would reject all cycles with fewer than $10$ nodes  or accept only knots which are completely connected subgraphs. 

\ \\
\textbf{Expiring links.} In the communication model, it is assumed that the sender process transfers its entire communication history to the receiver process across the communication link. This may require extensive communication and resources. 

Our algorithm may be adapted to limit resource usage. For example, the algorithm may discard the links older than some pre-determined period, say $P$. To put another way, the links and topological information expire after $P$ states. This model would nicely represent the network with moving topology or changing membership. 
In this case, the necessary conditions of Theorem~\ref{thrmIFF} must apply for the links within this period $P$.

 \ \\
 \textbf{Future research.} In this paper, we apply knot identification to the problem of agreement in dynamic networks. In the future, it would be interesting to study what other topological features can be effectively used for consensus and related tasks.
 Alternatively, it would be interesting to determine communication environments that naturally yield the dynamic graphs that comply with the adversary conditions allowing the solution the Knot Identification Problem. 

 Another promising research direction is implementing our knot identification algorithm in a complete system and testing its performance in practical environments such as Internet-of-Things networks. 

The computation model we consider can address message loss and process failure as special topologies. However, these faults are benign. It is interesting to address solvability of Knot 
 Identification and similar problems in the presence of Byzantine faults where faulty processes may behave arbitrarily~\cite{lamport1982byzantine,oglio2023consensus,pease1980reaching}.

\bibliographystyle{plain}
\bibliography{knot}

\end{document}